\documentclass[11pt]{article}

\usepackage[margin=1in]{geometry}
\usepackage{amsmath,amssymb,amsthm,bm}
\usepackage{booktabs,xcolor}
\usepackage{graphicx}
\usepackage{subcaption}
\usepackage{caption}
\usepackage{float}
\usepackage{algorithm}
\usepackage{algpseudocode}
\usepackage{natbib}
\usepackage[colorlinks=true,linkcolor=blue,citecolor=blue,urlcolor=blue]{hyperref}

\graphicspath{{figures/}{./}}

\DeclareMathOperator*{\argmin}{arg\,min}
\DeclareMathOperator{\col}{col}

\newcommand{\R}{\mathbb{R}}
\newcommand{\E}{\mathrm{E}}
\newcommand{\tr}{\mathrm{tr}}
\newcommand{\diag}{\mathrm{diag}}
\newcommand{\MN}{\mathcal{MN}}
\newcommand{\Ga}{\mathcal{G}}
\newcommand{\GIG}{\mathcal{GIG}}
\newcommand{\I}{\mathbf{I}}

\newcommand{\ExI}{\mathrm{ExI}}

\newtheorem{theorem}{Theorem}
\newtheorem{lemma}{Lemma}
\newtheorem{corollary}{Corollary}

\newcommand{\safeincludegraphics}[2][]{%
  \IfFileExists{#2}{\includegraphics[#1]{#2}}{%
    \fbox{\parbox[c][2.0in][c]{0.9\linewidth}{\centering Missing figure\\\texttt{\detokenize{#2}}}}%
  }%
}

\newcommand{\resetappendixcounters}{%
  \setcounter{equation}{0}%
  \renewcommand{\theequation}{\thesection.\arabic{equation}}%
  \renewcommand{\theHequation}{\thesection.\arabic{equation}}%
  \setcounter{theorem}{0}%
  \setcounter{lemma}{0}%
  \setcounter{corollary}{0}%
  \renewcommand{\thetheorem}{\thesection.\arabic{theorem}}%
  \renewcommand{\theHtheorem}{\thesection.\arabic{theorem}}%
  \renewcommand{\thelemma}{\thesection.\arabic{lemma}}%
  \renewcommand{\theHlemma}{\thesection.\arabic{lemma}}%
  \renewcommand{\thecorollary}{\thesection.\arabic{corollary}}%
  \renewcommand{\theHcorollary}{\thesection.\arabic{corollary}}%
  \setcounter{table}{0}%
  \setcounter{figure}{0}%
  \renewcommand{\thetable}{\thesection.\arabic{table}}%
  \renewcommand{\theHtable}{\thesection.\arabic{table}}%
  \renewcommand{\thefigure}{\thesection.\arabic{figure}}%
  \renewcommand{\theHfigure}{\thesection.\arabic{figure}}%
}

\title{Bayesian sparse principal coordinates analysis\\with delta-tolerant linear approximation for microbiome data}

\author{%
Hsin-Hsiung Huang$^{1}$, Ruitao Liu$^{2}$, Liangliang Zhang$^{2}$, and Shao-Hsuan Wang$^{3}$\\[0.75em]
\small $^{1}$School of Data, Mathematical, and Statistical Sciences, University of Central Florida, Orlando, Florida 32816, USA\\
\small $^{2}$Department of Population and Quantitative Health Sciences, Case Western Reserve University, Cleveland, Ohio 44106, USA\\
\small $^{3}$Graduate Institute of Statistics, National Central University, Taoyuan 320317, Taiwan\\
\small \texttt{hsin.huang@ucf.edu}, \texttt{picowang@gmail.com}
}
\date{}

\begin{document}
\maketitle

\begin{center}
\small This arXiv version combines the main manuscript and supplementary materials into a single preprint. The manuscript is currently under review at \emph{Biometrics}.
\end{center}

\begin{abstract}
Principal coordinates analysis (PCoA) is a standard exploratory tool for microbiome beta-diversity studies, but its axes are defined by pairwise dissimilarities and therefore do not directly identify the taxa driving an ordination. We propose Bayesian sparse principal coordinates analysis (BSPCoA), a post hoc framework that approximates the leading principal coordinates by a sparse linear surrogate in the observed taxa. A delta-tolerance diagnostic quantifies the discrepancy between the classical ordination and its best linear surrogate, clarifying when taxon-level interpretation is well supported. We place three-parameter beta normal global-local priors on the surrogate coefficients to induce row sparsity, obtain posterior uncertainty, and select influential taxa. The method reduces to sparse principal component analysis under Euclidean distance, while remaining applicable to ecologically meaningful dissimilarities such as Bray--Curtis and Hellinger distances. We conduct simulation studies to demonstrate that BSPCoA provides an approximately linear representation of the dominant ordination geometry while enhancing interpretability in sparse microbiome settings. In the Hadza gut microbiome data, the method produces an ordination close to that of classical PCoA while highlighting a parsimonious set of taxa associated with seasonal variation.
\end{abstract}


\section{Introduction}\label{sec:intro}

Principal coordinates analysis (PCoA), introduced by \citet{gower1966some}, is a standard exploratory tool for data that are most naturally summarized through pairwise distances or dissimilarities. It is especially attractive when scientifically meaningful dissimilarities are available, even if the original measurements are high-dimensional, sparse, or poorly represented by Euclidean covariance structure alone. In microbiome studies, PCoA is therefore a first-line method for visualizing beta diversity under distances such as Bray--Curtis and Hellinger, and the resulting low-dimensional coordinates often provide the first summary of between-sample structure.

A persistent limitation of classical PCoA is that the displayed axes are functions of the sample-to-sample dissimilarity matrix rather than explicit loadings on the original taxa. Consequently, while PCoA is well suited for visualization, it provides limited insight into which taxa are associated with the observed ordination. However, many studies require both an ordination respecting a chosen ecological distance and a sparse subset of taxa explaining the dominant geometric structure. Our goal is related to, but distinct from, constrained ordination methods such as distance-based redundancy analysis and canonical analysis of principal coordinates \citep{legendre1999,anderson2003}. Those approaches use external variables to explain or discriminate a distance matrix. By contrast, we begin with an \emph{unconstrained} PCoA and then construct a sparse \emph{post hoc} surrogate that attributes the leading principal coordinates to the original taxa. In this way, BSPCoA captures the dominant ordination geometry through an approximately linear representation while adding feature-level interpretation.

To this end, we build on the regression perspective underlying sparse principal component analysis (SPCA) \citep{zou2006sparse}. The leading PCoA coordinates are treated as pseudo-responses and approximated by a sparse linear surrogate in the original features. Related ideas arise in Bayesian multidimensional scaling and variable selection for distance-based embeddings \citep{oh2001,lin2018}, but these approaches do not directly target sparse taxon-level explanations for a fixed classical PCoA.
First, we propose a Bayesian sparse PCoA model using three-parameter beta--normal (TPBN) global--local shrinkage priors \citep{armagan2011generalized,carvalho2010}. Second, we introduce a $\delta$-tolerant diagnostic that quantifies how faithfully the leading principal coordinates can be represented by a linear surrogate. Based on this diagnostic, we establish posterior contraction results, following ideas similar to those of \citet{bai2018mbsp} and \citet{wang2025ejs}, to characterize row-sparse taxon effects and provide posterior uncertainty quantification.

The remainder of the paper is organized as follows. Section~\ref{sec:method} reviews the SPCA regression connection, develops the delta-tolerant surrogate for PCoA, and presents the Bayesian sparse PCoA algorithm. Section~\ref{sec:theory} gives an approximation-based posterior-concentration result, with formal assumptions and proofs deferred to Appendix~\ref{app:theory}. Section~\ref{sec:sim} presents Euclidean and microbiome-motivated simulations; additional Euclidean experiments and the runtime comparison are reported in Appendix~\ref{app:euclid}. Section~\ref{sec:app} contains the Hadza gut microbiome application. Section~\ref{sec:disc} concludes.

\section{Methodology}\label{sec:method}

\subsection{Review of the sparse PCA regression formulation}\label{subsec:spca}

Let $\mathbf{X}\in\R^{n\times p}$ be a centered data matrix with singular value decomposition
\begin{equation}
\mathbf{X}=\mathbf{U}\mathbf{D}\mathbf{V}^{\top},
\label{eq:svd}
\end{equation}
where $\mathbf{U}$ and $\mathbf{V}$ are orthogonal and $\mathbf{D}$ contains the singular values. The principal component score matrix is $\mathbf{Z}=\mathbf{U}\mathbf{D}$ and the columns of $\mathbf{V}$ are the corresponding loading vectors. For the $j$th principal component and any $\lambda>0$, \citet{zou2006sparse} considered the ridge regression problem
\begin{equation}
\widehat{\bm{\beta}}_{\mathrm{ridge}}
=
\argmin_{\bm{\beta}\in\R^p}
\left\{
\|\mathbf{Z}\bm{e}_j-\mathbf{X}\bm{\beta}\|_2^2+\lambda\|\bm{\beta}\|_2^2
\right\},
\label{eq:ridge}
\end{equation}
where $\bm{e}_j$ is the $j$th coordinate vector. The normalized solution of \eqref{eq:ridge} recovers the $j$th PCA loading vector. This observation extends to the first $k$ components. Let $\mathbf{A}=(\bm{\alpha}_1,\ldots,\bm{\alpha}_k)$ and $\mathbf{B}=(\bm{\beta}_1,\ldots,\bm{\beta}_k)$ be $p\times k$ matrices. Then the SPCA regression characterization of \citet{zou2006sparse} is as follows.

\begin{theorem}[Zou et al. (2006)]\label{thm:zou}
For $\lambda>0$, any minimizer of
\begin{equation}
(\widehat{\mathbf{A}},\widehat{\mathbf{B}})
=
\argmin_{\mathbf{A},\mathbf{B}}
\left\{
\|\mathbf{X}-\mathbf{X}\mathbf{B}\mathbf{A}^{\top}\|_F^2
+\lambda\sum_{j=1}^k\|\bm{\beta}_j\|_2^2
\right\}
\quad\text{subject to}\quad
\mathbf{A}^{\top}\mathbf{A}=\I_k
\label{eq:zou3}
\end{equation}
satisfies $\widehat{\bm{\beta}}_j\propto \bm{v}_j$ for $j=1,\ldots,k$, where $\bm{v}_j$ is the $j$th ordinary PCA loading vector.
\end{theorem}

SPCA then adds an $\ell_1$ or elastic-net penalty to \eqref{eq:zou3}. This regression viewpoint is useful here because the leading PCoA coordinates can also be treated as pseudo-responses. The difference is that, for PCoA, the target scores arise from a distance-based eigendecomposition rather than from the covariance structure of $\mathbf{X}$ itself.

\subsection{PCoA coordinates and a sparse linear surrogate}\label{subsec:pcoa_surrogate}

Let $\mathbf{X}=(\bm{x}_1,\ldots,\bm{x}_n)^{\top}\in\R^{n\times p}$ denote the feature matrix after any prespecified preprocessing used for interpretation, and let $d(\bm{x}_i,\bm{x}_{i'})$ be a dissimilarity between samples $i$ and $i'$. Classical PCoA is based on the double-centered similarity matrix
\begin{equation}
\mathbf{G}=-\frac{1}{2}\mathbf{J}\mathbf{D}^{(2)}\mathbf{J}\mbox{~~with~~}
\mathbf{J}=\mathbf{I}_n-\frac{1}{n}\mathbf{1}_n\mathbf{1}_n^{\top},
\label{eq:G}
\end{equation}
where $\mathbf{D}^{(2)}$ is the elementwise squared dissimilarity matrix. Let
$\mathbf{G}=\mathbf{H}\bm{\Lambda}\mathbf{H}^{\top}$
be an eigendecomposition with eigenvalues $\lambda_1\ge \cdots \ge \lambda_n$. When the dissimilarity is non-Euclidean, $\mathbf{G}$ can be indefinite, so we retain only the positive eigenvalues. Let $m=\sum_{r=1}^n I(\lambda_r>0)$, $\mathbf{H}_k=(\bm{h}_1,\ldots,\bm{h}_k)$, and $\bm{\Lambda}_k=\diag(\lambda_1,\ldots,\lambda_k)$ for some $k\le m$. The $n\times k$ matrix of leading principal coordinates is
\begin{equation}
\mathbf{Z}_k=\mathbf{H}_k\bm{\Lambda}_k^{1/2}.
\label{eq:Zk}
\end{equation}

BSPCoA is a post hoc surrogate for the classical ordination. Once $\mathbf{Z}_k$ is computed, we seek a coefficient matrix $\mathbf{B}\in\R^{p\times k}$ such that
\begin{equation}
\mathbf{Z}_k \approx \mathbf{X}\mathbf{B}.
\label{eq:linsur}
\end{equation}
Thus the surrogate coordinates $\mathbf{X}\mathbf{B}$ are directly interpretable in the original taxa. When the dissimilarity is Euclidean and the columns of $\mathbf{X}$ are centered, $\mathbf{G}=\mathbf{X}\mathbf{X}^{\top}$, so $\mathbf{Z}_k$ coincides with the leading principal component scores and the framework reduces to SPCA. A natural frequentist analogue is the following row-sparse regression problem
\begin{equation}
\widehat{\mathbf{B}}_{\mathrm{spcoa}}=
\argmin_{\mathbf{B}\in\mathbb{R}^{p\times k}}
\left\{
\|\mathbf{Z}_k-\mathbf{X}\mathbf{B}\|_F^2
+\xi\|\mathbf{B}\|_F^2
+\xi_1\sum_{j=1}^p\|\bm{b}_j\|_1
\right\},
\label{eq:spreg}
\end{equation}
where $\bm{b}_j^{\top}$ denotes the $j$th row of $\mathbf{B}$.
The penalty $\sum_{j=1}^p \|\bm{b}_j\|_1$ promotes joint selection of taxa across the retained ordination axes, thereby inducing a row-sparse structure consistent with the Bayesian formulation described below.

\subsection{Delta tolerance and explanation index}\label{subsec:delta}

To quantify how well a linear surrogate explains the leading PCoA geometry, we define
\begin{equation}
\delta(\mathbf{B})=
\frac{\|\mathbf{Z}_k-\mathbf{X}\mathbf{B}\|_F}{\|\mathbf{Z}_k\|_F}
\mbox{~~and~~}
\ExI(\mathbf{B})=
\frac{\tr\{(\mathbf{X}\mathbf{B})^{\top}\mathbf{Z}_k\}}
{\|\mathbf{X}\mathbf{B}\|_F\,\|\mathbf{Z}_k\|_F}.
\label{eq:exi_delta}
\end{equation}
The quantity $\delta(\mathbf{B})$ is a relative reconstruction error, whereas $\ExI(\mathbf{B})$ is the cosine of the angle between $\mathbf{X}\mathbf{B}$ and $\mathbf{Z}_k$ in Frobenius inner-product geometry. For a generic $\mathbf{B}$, $\delta(\mathbf{B})$ need not be bounded by one, but the best achievable value $\delta_{\star}=\min_{\mathbf{B}}\delta(\mathbf{B})$ always lies in $[0,1]$. Lemma~\ref{lem:bestlin} shows that $\delta_{\star}$ is the irreducible loss from forcing a linear surrogate onto an ordination defined through an arbitrary dissimilarity. In the Euclidean case with centered $\mathbf{X}$, the leading PCA loading matrix $\mathbf{V}_k$ satisfies $\mathbf{Z}_k=\mathbf{X}\mathbf{V}_k$, and therefore $\delta_{\star}=0$.

\begin{lemma}\label{lem:bestlin}
Let $\mathbf{X}^{+}$ denote the Moore--Penrose pseudoinverse of $\mathbf{X}$ and let $\mathbf{P}_{X}=\mathbf{X}\mathbf{X}^{+}$ be the orthogonal projector onto $\col(\mathbf{X})$. Then any minimizer of $\min_{\mathbf{B}}\|\mathbf{Z}_k-\mathbf{X}\mathbf{B}\|_F^2$ satisfies $\mathbf{X}\mathbf{B}_{\star}=\mathbf{P}_{X}\mathbf{Z}_k$; one convenient choice is $\mathbf{B}_{\star}=\mathbf{X}^{+}\mathbf{Z}_k$. Moreover,
\begin{equation}
\delta_{\star}^2=
\frac{\|(\mathbf{I}_n-\mathbf{P}_X)\mathbf{Z}_k\|_F^2}{\|\mathbf{Z}_k\|_F^2}
=
1-\frac{\|\mathbf{P}_X\mathbf{Z}_k\|_F^2}{\|\mathbf{Z}_k\|_F^2},
\label{eq:deltastar}
\end{equation}
and $\ExI(\mathbf{B}_{\star})^2=1-\delta_{\star}^2$.
\end{lemma}
\begin{proof}
The proof is deferred to Appendix~\ref{app:theory}.
\end{proof}

\subsection{Bayesian sparse PCoA}\label{subsec:bayes}

We now place a Bayesian row-sparse model on the surrogate coefficients. Given the pseudo-response matrix $\mathbf{Y}=\mathbf{Z}_k\in\R^{n\times k}$, we consider the working multivariate regression model
\begin{equation}
\mathbf{Y}=\mathbf{X}\mathbf{B}+\mathbf{E},
\qquad
\mathbf{E}\sim \MN_{n\times k}(\mathbf{0},\mathbf{I}_n,\mathbf{I}_k).
\label{eq:working}
\end{equation}
Because $\mathbf{Y}$ is a deterministic function of the observed distance matrix, model \eqref{eq:working} should be interpreted as a working model that induces shrinkage and uncertainty quantification for the linear surrogate, rather than as a generative model for the original microbiome counts.

To encourage row sparsity, we place a shared local scale on each row of $\mathbf{B}=(\bm{b}_1,\ldots,\bm{b}_p)^{\top}$:
\begin{align}
\mathbf{B}\mid \psi_1,\ldots,\psi_p
&\sim \MN_{p\times k}\!\left(\mathbf{0},\diag(\psi_1,\ldots,\psi_p),\mathbf{I}_k\right), \notag\\
\psi_j\mid \zeta_j &\overset{\mathrm{ind}}{\sim} \Ga(u,\zeta_j), \qquad
\zeta_j \overset{\mathrm{iid}}{\sim} \Ga(a,\tau), \qquad j=1,\ldots,p .
\label{eq:tpbn}
\end{align}
This TPBN hierarchy includes the horseshoe prior as the special case $u=a=1/2$ \citep{carvalho2010,armagan2011generalized}. In the multivariate regression setting, related TPBN constructions and their theoretical guarantees have been studied by \citet{bai2018mbsp} and extended to mixed-type multivariate responses by \citet{wang2025ejs}. Throughout the numerical experiments, we adopt the horseshoe specification together with $\tau=(pn\log n)^{-1}$.

Conditional on $(\psi_1,\ldots,\psi_p)$, the posterior update for $\mathbf{B}$ is matrix normal:
\[
\mathbf{B}\mid \mathbf{Y},\mathbf{X},\psi_1,\ldots,\psi_p
\sim
\MN_{p\times k}(\mathbf{M}_B,\mathbf{V}_B,\mathbf{I}_k),
\]
where
\[
\mathbf{V}_B=(\mathbf{X}^{\top}\mathbf{X}+\bm{\Psi}^{-1})^{-1},
\qquad
\mathbf{M}_B=\mathbf{V}_B\mathbf{X}^{\top}\mathbf{Y},
\qquad
\bm{\Psi}=\diag(\psi_1,\ldots,\psi_p).
\]
For each row $j$,
\[
\psi_j\mid \zeta_j,\bm{b}_j
\sim
\GIG\!\left(u-\frac{k}{2},\ \|\bm{b}_j\|_2^2,\ 2\zeta_j\right),
\qquad
\zeta_j\mid \psi_j
\sim
\Ga(a+u,\tau+\psi_j).
\]

Under the working model \eqref{eq:working} and the TPBN prior \eqref{eq:tpbn},
the proposed Bayesian procedure is equivalent to solving the penalized
least-squares problem
\begin{equation}
\widetilde{\mathbf{B}}_{\mathrm{bspcoa}}
=
\argmin_{\mathbf{B}\in\mathbb{R}^{p\times k}}
\|\mathbf{Z}_k-\mathbf{X}\mathbf{B}\|_F^2
+ g(\mathbf{B}),
\label{eq:bspcoa_penalized}
\end{equation}
where $g(\mathbf{B})$ denotes the penalty function induced by the prior
distribution on $\mathbf{B}$. We further consider a two-stage exploratory procedure: first performing PCA and then
seeking a suitable sparse approximation. Specifically, motivated by
Theorem~\ref{thm:zou}, we consider the optimization problem
\begin{equation}
(\widehat{\mathbf{A}}_{\mathrm{bspcoa}},
\widehat{\mathbf{B}}_{\mathrm{bspcoa}})
=
\argmin_{\mathbf{A}\in\mathbb{R}^{n\times k},~\mathbf{B}\in\mathbb{R}^{p\times k}}
\|\mathbf{\Psi}_c^{1/2}-\mathbf{X}\mathbf{B}\mathbf{A}^\top\|_F^2
+ g(\mathbf{B})
\quad
\text{subject to}\quad
\mathbf{A}^\top\mathbf{A}=\mathbf{I}_k.
\label{eq:bspcoa_joint_opt}
\end{equation}
The resulting estimator admits Algorithm~\ref{alg:bspcoa}.

\begin{algorithm}[H]
\caption{Bayesian sparse PCoA}
\label{alg:bspcoa}
\begin{algorithmic}[1]
\Require Feature matrix $\mathbf{X}\in\R^{n\times p}$, target dimension $k$, dissimilarity $d(\cdot,\cdot)$, hyperparameters $(a,u,\tau)$, and number of MCMC iterations $T$
\Ensure Posterior summary $\widehat{\mathbf{B}}$ and fitted sparse coordinates $\mathbf{X}\widehat{\mathbf{B}}$ and $\widehat{\mathbf{A}}$
\State \textbf{Step 1:} Based on data matrix $\mathbf{X}$, choose a specific dissimilarity distance to obtain an $n\times n$ similarity matrix $\mathbf{\Psi}_c$. Take the leading $k$ eigenvectors of $\mathbf{\Psi}_c$ as the initial value of $\mathbf{A}$.
\State \textbf{Step 2:} Given $\mathbf{A}$, estimate $\mathbf{B}$ based on a linear model:
\State \hspace{1em}\textbf{Step 2.1:} Generate a Gaussian random matrix $\mathbf{E}\sim {\cal MN}_{n\times k}(\mathbf{O}, \mathbf{I}_n, \mathbf{I}_k)$ and set $\mathbf{Y}= \mathbf{\Psi}_c^{1/2} \mathbf{A}+\mathbf{E}$.
\State \hspace{1em}\textbf{Step 2.2:} Implement Gibbs sampling as below.
\begin{itemize}
    \item[] $\mathbf{B} \mid \psi_1, \ldots, \psi_p, \mathbf{Y}, \mathbf{X} \sim {\cal MN}_{p\times k}\!\left(\left(\mathbf{X}^\top \mathbf{X} + \mathbf{\Sigma}_{\mathbf{\Psi}}^{-1}\right)^{-1} \mathbf{X}^\top\mathbf{Y}, \left(\mathbf{X}^\top \mathbf{X} + \mathbf{\Sigma}_{\mathbf{\Psi}}^{-1}\right)^{-1}, \mathbf{I}_k\right)$, where $\mathbf{\Sigma}_{\mathbf{\Psi}}^{-1} = {\rm diag}( \psi_1^{-1}, \ldots,\psi_p^{-1})$,
    \item[] $\psi_j \mid \zeta_j, \mathbf{B} \sim \mathcal{GIG}\left(2 \zeta_j, (\mathbf{B} \mathbf{B}^\top)_{jj}, -\frac{q}{2} + u\right)$,~~$j=1,\dots,p$,
    \item[] $\zeta_j \mid \psi_j \sim \mathcal{G}(a + u, \tau + \psi_j)$,~~$j=1,\dots,p$.
\end{itemize}
\State \hspace{1em}\textbf{Step 2.3:} Update $\mathbf{B}$ as the MAP based on $95\%$ credible intervals after discarding the burn-in sample.
\State \textbf{Step 3:} Given $\mathbf{B}$, estimate $\mathbf{A}$ by deriving a solution to the Procrustes criterion: compute the singular value decomposition $\mathbf{\Psi}_c^{1/2}\mathbf{X} \mathbf{B}={\mathbb U}{\mathbb D}{\mathbb V}^\top$, then update $\mathbf{A}={\mathbb U}_{[,1:k]}{\mathbb V}^\top$.
\State \textbf{Step 4:} Repeat Steps 2--3 until convergence and obtain estimates of $(\mathbf{A}, \mathbf{B})$.
\end{algorithmic}
\end{algorithm}

\section{Posterior concentration for $\delta$-tolerant surrogates}\label{sec:theory}

To assess whether BSPCoA is able to recover the dominant trend of the PCoA coordinates, we adopt an approximation-based viewpoint. Rather than requiring the leading PCoA coordinates to admit an exact linear representation in the original feature space, we consider a set of linear surrogates that approximate the ordination geometry within a prescribed tolerance level. Formally, define
\[
\mathcal{B}_0^{\delta}=
\left\{
\mathbf{B}_0\in\mathbb{R}^{p\times k}:
\|\mathbf{Z}_k-\mathbf{X}\mathbf{B}_0\|_F
\le
\delta\|\mathbf{Z}_k\|_F
\right\}.
\]

Elements of $\mathcal{B}_0^{\delta}$ are called \emph{$\delta$-tolerant linear surrogates} for the leading PCoA coordinates. When $\delta$ is small, the dominant distance-based geometry can be well approximated by a sparse linear representation in the original taxa space. This notion differs from classical consistency, which typically assumes that the true parameter lies exactly in the model class. In contrast, $\mathcal{B}_0^{\delta}$ allows for controlled approximation error, reflecting the fact that the PCoA coordinates may not admit an exact sparse linear representation, especially when the dissimilarity measure is non-Euclidean. Consequently, our theoretical analysis focuses on posterior concentration around the set of approximate surrogates $\mathcal{B}_0^{\delta}$. We work under the same high-dimensional design conditions and TPBN prior regularity assumptions used by \citet{bai2018mbsp} and \citet{wang2025ejs}; the formal assumptions and proofs are given in Appendix~\ref{app:theory}. Throughout this section, the expectation $\E_{\mathbf{B}_0}=\E(\cdot\mid \mathbf{B}_0)$ and posterior ${\rm P}(\cdot\mid \mathbf{Y},\mathbf{X})$ are taken with respect to the working model in \eqref{eq:working}.

\begin{theorem}\label{thm:posterior}
Suppose the design and prior conditions stated in Appendix~\ref{app:theory} hold and $\|\mathbf{Z}_k\|_F^2=O(n)$. Then for every $\delta>0$ there exist positive constants $c$ and $C$ such that, for any $\varepsilon>c\delta$,
\[
\sup_{\mathbf{B}_0\in\mathcal{B}_0^{\delta}}
\E_{\mathbf{B}_0}\left [
{\rm P}
\!\left(
\mathbf{B}:
\|\mathbf{B}-\mathbf{B}_0\|_F>\varepsilon
\;\middle|\;
\mathbf{Y}=\mathbf{Z}_k,\mathbf{X}
\right)\right]
\le
\exp(-Cn\varepsilon^2)
\]
for all sufficiently large $n$.
\end{theorem}

Theorem~\ref{thm:posterior} shows that the posterior need not concentrate around an exact representation of the PCoA coordinates, which may not exist under a non-Euclidean dissimilarity. Instead, the posterior concentrates around the best available linear surrogate up to the tolerance level $\delta$. The condition $\varepsilon>c\delta$ explicitly separates statistical estimation error from surrogate approximation error.

\begin{corollary}\label{cor:euclid}
Under Euclidean distance with centered $\mathbf{X}$, suppose $\log p=o(n)$ and the conditions of Theorem~\ref{thm:posterior} hold. Then for any $\varepsilon>0$,
\[
\sup_{\mathbf{B}_0\in\mathcal{B}_0^0}
\E_{\mathbf{B}_0}
\left[{\rm P}
\!\left(
\mathbf{B}:
\|\mathbf{B}-\mathbf{B}_0\|_F
>
\varepsilon
\left(\frac{\log p}{n}\right)^{1/2}
\;\middle|\;
\mathbf{Y}=\mathbf{Z}_k,\mathbf{X}
\right)\right]
\longrightarrow 0
\qquad
\text{as }
\min\{n,p\}\to\infty.
\]
\end{corollary}

Corollary~\ref{cor:euclid} shows that BSPCoA reduces to the familiar high-dimensional SPCA regime when the dissimilarity is generated by ordinary Euclidean geometry.

\section{Simulation study}\label{sec:sim}
In this section, we evaluate the performance of the proposed method in comparison with traditional approaches.
Section~\ref{subsec:euclid_sim} compares SPCA with BSPCA, which can be viewed as a reduced version of BSPCoA.
Section~\ref{subsec:bray} generates synthetic microbiome count data from a Dirichlet--multinomial model designed
to mimic the variability commonly observed in sequencing experiments, in order to compare PCoA and BSPCoA.
Finally, Section~\ref{subsec:large_n} extends the previous setting to illustrate the computational advantage of BSPCoA over PCoA in large-sample scenarios. The code can be downloaded from the following link:
\href{https://github.com/picowang1203/Bayesian-PCoA}{https://github.com/picowang1203/Bayesian-PCoA}.

\subsection{Euclidean benchmark}\label{subsec:euclid_sim}

We first consider a Euclidean benchmark similar to the simulation settings in \citet{guan2009} and \citet{zou2006sparse}. Under Euclidean distance, BSPCoA reduces to Bayesian SPCA, so this experiment serves as a sanity check against the SPCA literature. Additional Euclidean loading tables are reported in Appendix~\ref{app:euclid} (Tables~\ref{tab:euclid_small}--\ref{tab:euclid_large}). Suppose that $n$ observations of $p$ variables are collected. The first $p_1$ variables are associated with a latent factor $V_1$, the next $p_2-p_1$ variables are associated with another latent factor $V_2$, and the remaining variables contain only noise. The latent factors are independent and generated as $V_1 \sim N(0,10)$ and $V_2 \sim N(0,20)$. Given $(V_1,V_2)$, the observed variables are generated according to
\begin{align*}
X_i &= V_1 + \varepsilon_i, \qquad && i=1,\ldots,p_1,\\
X_i &= V_2 + \varepsilon_i, \qquad && i=p_1+1,\ldots,p_2,\\
X_i &= \varepsilon_i, \qquad && i=p_2+1,\ldots,p,
\end{align*}
where $\varepsilon_i\overset{\mathrm{iid}}{\sim}N(0,1)$. Across both the low-dimensional and high-dimensional settings, BSPCoA recovers the same sparse support as frequentist SPCA up to the usual componentwise sign indeterminacy, while preserving nearly the same explained variation; see Appendix~\ref{app:euclid}.

\subsection{Simulation model for microbiome count data}\label{subsec:bray}

We generate synthetic microbiome count data from a Dirichlet--multinomial model designed to mimic the variability commonly observed in sequencing experiments. Consider two groups of samples, denoted by Group A and Group B, containing $n_A$ and $n_B$ samples, respectively. Each sample consists of counts for $p$ operational taxonomic units (OTUs). To reflect variability in sequencing depth, the total read count (library size) of each sample is generated from a Poisson distribution,
\[
L_i^{(A)} \sim \mathrm{Poisson}(8000), \qquad i=1,\dots,n_A
\quad\text{and}\quad
L_i^{(B)} \sim \mathrm{Poisson}(8000), \qquad i=1,\dots,n_B.
\]
For each sample, the underlying $p \times 1$ OTU composition vector is drawn from a Dirichlet distribution. Specifically,
\[
\mathbf p_i^{(A)}
\sim
\mathrm{Dirichlet}(\boldsymbol{\alpha}_A),
\qquad i=1,\dots,n_A
\quad\text{and}\quad
\mathbf p_i^{(B)}
\sim
\mathrm{Dirichlet}(\boldsymbol{\alpha}_B),
\qquad i=1,\dots,n_B,
\]
where
\[
\boldsymbol{\alpha}_A=
(6,6,6,6,6,2,2,2,2,2,\alpha,\dots,\alpha)\quad\text{and}\quad
\boldsymbol{\alpha}_B
=
(2,2,2,2,2,6,6,6,6,6,\alpha,\dots,\alpha),
\]
denote the Dirichlet concentration parameters for the two groups. The parameter $\alpha$ controls the abundance of the remaining OTUs and therefore determines the overall sparsity level. Conditional on the library size and the OTU composition, the observed $p\times1$ count vector for sample $i$ is generated from a multinomial distribution,
\begin{align*}
\mathbf X_i^{(A)}
\sim
\mathrm{Multinomial}\!\left(L_i^{(A)},\mathbf p_i^{(A)}\right)\mbox{~~and~~}\mathbf X_i^{(B)}
\sim
\mathrm{Multinomial}\!\left(L_i^{(B)},\mathbf p_i^{(B)}\right)
\end{align*}
and the full data matrix is given by $\mathbf X=
(
\mathbf X_1^{(A)},\dots,
\mathbf X_{n_A}^{(A)},
\mathbf X_1^{(B)},\dots,
\mathbf X_{n_B}^{(B)}
)^\top $.

We consider two settings, $\alpha=0.5$ and $\alpha=0.1$, corresponding to relatively dense and sparse OTU compositions, respectively. Pairwise dissimilarities are computed using the Bray--Curtis distance
\[
d_{ii'}=
\frac{\sum_{k=1}^p|x_{ik}-x_{i'k}|}
{\sum_{k=1}^p(x_{ik}+x_{i'k})}.
\]
We then compare classical PCoA and BSPCoA based on the first two fitted coordinates. For BSPCoA, we report $\delta_{\mathrm{res}}=\delta(\widehat{\mathbf{B}})$ and $\ExI=\ExI(\widehat{\mathbf{B}})$; unlike $\delta_\star$, the residual quantity $\delta_{\mathrm{res}}$ can exceed one because it is evaluated at the fitted sparse surrogate rather than at the unrestricted best linear approximation. In addition, we report the proportion of positive-eigenvalue mass explained by the first two classical PCoA coordinates, the average silhouette width computed from the two-dimensional embedding, and the best-matched clustering accuracy,
\[
\mathrm{BM\text{-}ACC}=
\max_{\pi\in\mathcal P}
\frac{1}{n}\sum_{i=1}^n\mathbf{1}\{y_i=\pi(\widehat c_i)\},
\]
where $y_i$ denotes the true group label, $\widehat c_i$ is the cluster assignment obtained from $k$-means with $k=2$, and $\mathcal P$ denotes the set of permutations of the cluster labels.

Table~\ref{tab:alpha1} summarizes the baseline scenario without background perturbation among the nondifferential taxa. In this setting, BSPCoA preserves the dominant ordination structure while providing an interpretable sparse surrogate. The clustering performance is favorable, as indicated by larger silhouette widths and comparable BM--ACC values relative to classical PCoA. Table~\ref{tab:alpha2} reports a more challenging scenario in which weak background perturbations are introduced to the nondifferential taxa, increasing background heterogeneity without altering the primary group signal. Even under this setting, BSPCoA continues to produce reliable clustering results. Although BSPCoA explains a smaller proportion of positive-eigenvalue mass, it still achieves strong clustering accuracy. This suggests that the Bayesian sparse formulation can recover the underlying group structure while relying on a parsimonious linear explanation of the ordination geometry.

\subsection{Large-sample implementation}\label{subsec:large_n}

The main computational bottleneck of classical PCoA is the formation and eigendecomposition of the $n\times n$ matrix $\mathbf{G}$. When $n$ is large, we estimate the surrogate on a representative subsample $I_m\subset\{1,\ldots,n\}$ of size $m\ll n$. Specifically, we compute $\mathbf{Z}_{k}^{(m)}$ and fit $\widehat{\mathbf{B}}^{(m)}$ using only the subsample $(\mathbf{X}_{I_m,\cdot},\mathbf{Z}_{k}^{(m)})$, after which all $n$ observations are embedded by the linear projection $\mathbf{X}\widehat{\mathbf{B}}^{(m)}$. This strategy preserves taxon-level interpretation while avoiding an eigendecomposition of the full $n\times n$ distance matrix. A runtime comparison for this subsample projection strategy is reported in Appendix~\ref{app:euclid}.

\section{Hadza gut microbiome dataset}\label{sec:app}

We applied BSPCoA to the Hadza gut microbiome study of \citet{smits2017seasonal}, which reported recurrent seasonal changes in community composition. We focus on the wet-versus-dry contrast and ask whether BSPCoA can recover an ordination comparable to classical PCoA while identifying a sparse set of taxa associated with dominant seasonal variation.

\subsection{Preprocessing and model fitting}\label{subsec:hadza_pre}

We restricted attention to adult participants sampled in 2014 to reduce demographic and temporal heterogeneity. Taxa were filtered by prevalence to reduce extreme sparsity, counts were converted to relative abundances, and the Hellinger distance
\begin{equation}
d_H(i,i')=
\left\{
\sum_{j=1}^p
\left(
\sqrt{p_{ij}}-\sqrt{p_{i'j}}
\right)^2
\right\}^{1/2}
\label{eq:hellinger}
\end{equation}
was computed from the transformed compositions. The Hellinger distance is attractive in community ecology because it corresponds to Euclidean distance after square-root transformation \citep{legendre2001}. We then fit BSPCoA with target dimension $k=2$, the horseshoe TPBN prior $u=a=1/2$, and $\tau=(pn\log n)^{-1}$. Posterior summaries were based on 1500 retained Gibbs draws after discarding the first 500 iterations as burn-in.

For the classical PCoA, the first two axes explain 13.2\% and 8.0\% of the total positive-eigenvalue mass, respectively. Figure~\ref{fig:hadza_ord} compares the classical ordination with the two-dimensional BSPCoA surrogate. In both displays, wet and dry samples overlap, but the BSPCoA embedding is somewhat more concentrated and slightly more aligned with the seasonal contrast. To quantify this alignment, we applied $k$-means with $k=2$ to the two-dimensional coordinates. The best-matched clustering accuracy increases from 0.559 for classical PCoA to 0.609 for BSPCoA, suggesting that shrinkage removes some extraneous variation while preserving the dominant distance-based pattern. Figure~\ref{fig:hadza_heatmap} displays a heatmap of the posterior median loading matrix $\widehat{\mathbf{B}}$. For display only, we rescale the loading matrix elementwise to the unit interval by
\[
\widetilde{\mathbf{B}}=
\frac{\widehat{\mathbf{B}}-\min(\widehat{\mathbf{B}})}
{\max(\widehat{\mathbf{B}})-\min(\widehat{\mathbf{B}})}.
\]
Unlike classical PCoA, BSPCoA yields taxon-level coefficients, so the heatmap identifies a relatively small subset of taxa with the strongest contributions to the first two sparse ordination directions.

\section{Discussion}\label{sec:disc}

We proposed Bayesian sparse principal coordinates analysis (BSPCoA), a post hoc framework that links distance-based ordination with interpretable taxon-level effects through a sparse linear surrogate. By approximating the leading PCoA coordinates using a row-sparse regression model, the proposed method preserves the dominant geometry of the classical ordination while providing direct interpretation in terms of the original taxa. The introduction of the $\delta$-tolerance diagnostic clarifies when a sparse linear surrogate can faithfully represent the distance-based embedding and offers a quantitative measure of approximation quality. The simulation studies demonstrate that BSPCoA maintains clustering performance comparable to classical PCoA while improving interpretability through sparse feature selection. In particular, the results show that reliable classification can be achieved even when the surrogate explains a relatively small proportion of the total eigenvalue mass. The Euclidean benchmark further confirms that the proposed framework reduces naturally to Bayesian sparse PCA when the dissimilarity corresponds to ordinary Euclidean geometry.

The analysis of the Hadza gut microbiome data illustrates that BSPCoA produces an ordination close to classical PCoA while highlighting a small set of taxa associated with seasonal variation. Nevertheless, several limitations deserve attention. The interpretability of the surrogate depends on the chosen preprocessing and dissimilarity, and when the optimal $\delta$ is large, the linear surrogate may provide only a weak explanation of the distance-based geometry. Future work may consider structured shrinkage informed by phylogenetic relationships, alternative covariance structures across ordination axes, and extensions to longitudinal or multi-view microbiome studies.

\section*{Acknowledgements}
The authors thank the editor, associate editor, and reviewers for their constructive comments. This work was partially supported by the National Science Foundation under grants DMS-1924792 and DMS-2318925.

\section*{Data availability}
The Hadza microbiome data analyzed in this article are publicly available from the study of \citet{smits2017seasonal} and the repository referenced therein. Code and processed data sufficient to reproduce the analyses should be deposited in \href{https://github.com/picowang1203/Bayesian-PCoA}{https://github.com/picowang1203/Bayesian-PCoA}.

\clearpage

\begin{table}[p]
\centering
\caption{Simulation results comparing PCoA and BSPCoA under the baseline Dirichlet--multinomial setting with $\alpha=0.5$.}
\label{tab:alpha1}
\begin{tabular}{llrrrrrr}
\toprule
Method &  & var.PC1 (\%) & var.PC2 (\%) & Silhouette$_{2D}$ & $\delta_{\mathrm{res}}$ & BM-ACC & ExI \\
\midrule
PCoA   & mean & 37.40 & 2.22 & 0.53 & NA      & 0.99 & NA \\
       & sd   & 0.805 & 0.063 & 0.005 & NA      & 0.004 & NA \\
\midrule
BSPCoA & mean & 16.72 & 10.60 & 0.61 & 1.98 & 0.98 & 0.401 \\
       & sd   & 0.716 & 0.069 & 0.148 & 0.556 & 0.061 & 0.058 \\
\bottomrule
\end{tabular}
\end{table}

\begin{table}[p]
\centering
\caption{Simulation results comparing PCoA and BSPCoA under the sparse Dirichlet--multinomial setting with $\alpha=0.1$.}
\label{tab:alpha2}
\begin{tabular}{llrrrrrr}
\toprule
Method &  & var.PC1 (\%) & var.PC2 (\%) & Silhouette$_{2D}$ & $\delta_{\mathrm{res}}$ & BM-ACC & ExI \\
\midrule
PCoA   & mean & 65.4 & 2.02 & 0.56 & NA      & 1.00 & NA \\
       & sd   & 1.011 & 0.120 & 0.010 & NA      & 0.000 & NA \\
\midrule
BSPCoA & mean & 33.3 & 12.1 & 0.75 & 1.32 & 1.00 & 0.411 \\
       & sd   & 0.852 & 0.081 & 0.143 & 0.264 & 0.000 & 0.0681 \\
\bottomrule
\end{tabular}
\end{table}

\begin{figure}[p]
\centering
\safeincludegraphics[width=\textwidth]{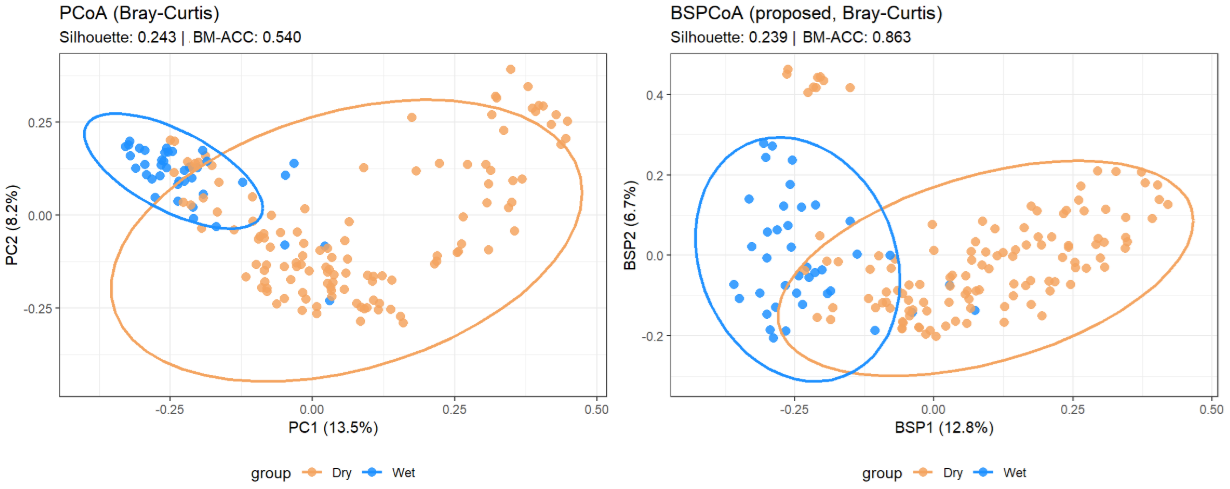}
\caption{Classical PCoA and BSPCoA for the Hadza gut microbiome data under the Hellinger distance. The displayed points are the first two coordinates from each method, with wet and dry season samples indicated by different plotting symbols.}
\label{fig:hadza_ord}
\end{figure}

\begin{figure}[p]
\centering
\safeincludegraphics[width=.8\textwidth]{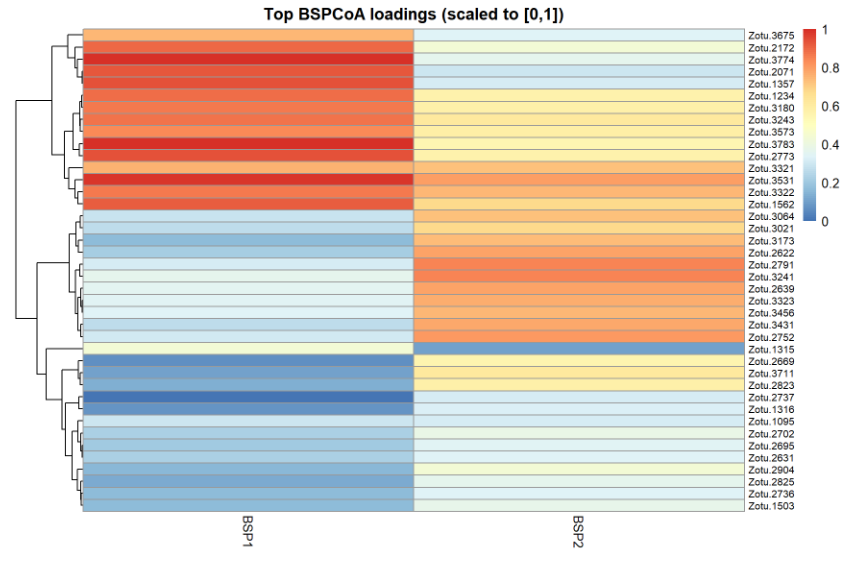}
\caption{Heatmap of the posterior median BSPCoA loading matrix for the Hadza data. Larger magnitudes indicate taxa with stronger contributions to the sparse ordination directions.}
\label{fig:hadza_heatmap}
\end{figure}

\appendix

\section{Technical assumptions and proofs}\label{app:theory}
\resetappendixcounters

We first record the assumptions used in the posterior concentration result. These conditions are adapted from the multivariate regression theory of \citet{bai2018high} and specialized to the pseudo-response $\mathbf{Y}=\mathbf{Z}_k$.

\paragraph{Assumption A.}
\begin{itemize}
\item[(A1)] Either $p=p_n>n$ for all $n$ with $\log p_n = O(n^\alpha)$ for some $\alpha\in(0,1)$, or $p_n=o(n)$ with $p_n\le n$ for all $n$.
\item[(A2)] There exists a constant $c_1>0$ such that, for every subset $S\subset\{1,\ldots,p_n\}$ with $|S|\le n$,
\[
 \frac{1}{c_1} \le\liminf_n \lambda_{\min}\!\left(\frac{\mathbf{X}_S^\top \mathbf{X}_S}{n}\right)\mbox{~~and~~}
\limsup_n \lambda_{\max}\!\left(\frac{\mathbf{X}^\top \mathbf{X}}{n}\right)\le c_1
\]
and the active set size satisfies $s_0=o(n/\log p)$.
\end{itemize}

\paragraph{Assumption B.}
\begin{itemize}
\item[(B1)] The slowly varying function $L(t)$ associated with the TPBN prior tail is positive and bounded away from zero for large $t$.
\item[(B2)] All entries of every $\mathbf{B}_0\in \mathcal{B}_0^\delta$ are uniformly bounded.
\item[(B3)] The global shrinkage parameter satisfies $\tau_n\in(0,1)$ for all $n$ and $\tau_n=o(p^{-1}n^{-\rho})$ for some $\rho>0$.
\item[(B4)] The leading PCoA coordinate matrix obeys $\|\mathbf{Z}_k\|_F^2=O(n)$.
\end{itemize}

\begin{proof}[Proof of Lemma~\ref{lem:bestlin}.]
For every $\mathbf{B}$, the matrix $\mathbf{X}\mathbf{B}$ belongs to $\mathrm{col}(\mathbf{X})$. Decompose
\[
\mathbf{Z}_k-\mathbf{X}\mathbf{B}
=
(\mathbf{I}_n-\mathbf{P}_X)\mathbf{Z}_k
+
\{\mathbf{P}_X\mathbf{Z}_k-\mathbf{X}\mathbf{B}\}.
\]
The first term lies in $\mathrm{col}(\mathbf{X})^\perp$ and the second lies in $\mathrm{col}(\mathbf{X})$, so the two terms are orthogonal. Hence
\[
\|\mathbf{Z}_k-\mathbf{X}\mathbf{B}\|_F^2
=
\|(\mathbf{I}_n-\mathbf{P}_X)\mathbf{Z}_k\|_F^2
+
\|\mathbf{P}_X\mathbf{Z}_k-\mathbf{X}\mathbf{B}\|_F^2.
\]
The minimum is achieved if and only if $\mathbf{X}\mathbf{B}=\mathbf{P}_X\mathbf{Z}_k$, and $\mathbf{B}_\star=\mathbf{X}^+\mathbf{Z}_k$ is one such minimizer. Dividing by $\|\mathbf{Z}_k\|_F^2$ yields the expression for $\delta_\star^2$. Finally,
\[
\tr\{(\mathbf{X}\mathbf{B}_\star)^\top\mathbf{Z}_k\}
=
\tr(\mathbf{Z}_k^\top \mathbf{P}_X \mathbf{Z}_k)
=
\|\mathbf{P}_X\mathbf{Z}_k\|_F^2,
\]
which gives $\mathrm{ExI}(\mathbf{B}_\star)=\|\mathbf{P}_X\mathbf{Z}_k\|_F/\|\mathbf{Z}_k\|_F$ and therefore $\mathrm{ExI}(\mathbf{B}_\star)^2=1-\delta_\star^2$.
\end{proof}

\begin{proof}[Proof of Theorem~\ref{thm:posterior}.]
Fix $\mathbf{B}_0\in\mathcal{B}_0^\delta$ and define the deterministic approximation error
\[
\mathbf{R}_0=\mathbf{Z}_k-\mathbf{X}\mathbf{B}_0,
\qquad
\|\mathbf{R}_0\|_F\le \delta\|\mathbf{Z}_k\|_F=O(\delta\sqrt{n})
\]
by Assumption~(B4). Under the working model, the pseudo-response can be written as
\[
\mathbf{Y}
=
\mathbf{Z}_k+\mathbf{E}
=
\mathbf{X}\mathbf{B}_0+\mathbf{R}_0+\mathbf{E},
\qquad
\mathbf{E}\sim \MN_{n\times k}(\mathbf{0},\mathbf{I}_n,\mathbf{I}_k).
\]
Thus the usual exact multivariate regression model is perturbed by the deterministic bias $\mathbf{R}_0$.

The posterior-contraction proof of \citet{bai2018high} is based on exponentially consistent tests and lower bounds on prior mass for shrinking neighborhoods of $\mathbf{B}_0$. The prior-mass part is unchanged here because it depends only on the TPBN prior and the Euclidean neighborhood of $\mathbf{B}_0$, not on the pseudo-response construction. Hence it remains to show that the same tests are stable under the additional bias $\mathbf{R}_0$.

Let $S\subset\{1,\ldots,p\}$ be any candidate sparse model with $|S|\le n$, and let $\mathbf{X}_S$ denote the corresponding design submatrix. By Assumption~(A2),
\begin{equation}
\left\|
(\mathbf{X}_S^\top \mathbf{X}_S)^{-1}\mathbf{X}_S^\top
\right\|_{\mathrm{op}}
\le
C_1 n^{-1/2}
\label{eq:projbound}
\end{equation}
uniformly over such $S$. Therefore
\begin{equation}
\left\|
(\mathbf{X}_S^\top \mathbf{X}_S)^{-1}\mathbf{X}_S^\top \mathbf{R}_0
\right\|_F
\le
\left\|
(\mathbf{X}_S^\top \mathbf{X}_S)^{-1}\mathbf{X}_S^\top
\right\|_{\mathrm{op}}
\|\mathbf{R}_0\|_F
\le
C_2 \delta
\label{eq:biasbound}
\end{equation}
for some constant $C_2>0$. Equation~\eqref{eq:biasbound} is the only place where the pseudo-response setting differs from the exact regression model.

Let $\Phi_n$ be the exponentially consistent test used in the proof of Theorem~2 of \citet{bai2018high} for the exact model $\mathbf{Y}^\dagger=\mathbf{X}\mathbf{B}_0+\mathbf{E}$. Replacing $\mathbf{Y}^\dagger$ by $\mathbf{Y}=\mathbf{Y}^\dagger+\mathbf{R}_0$ shifts the relevant least-squares contrasts by at most $C_2\delta$ uniformly over the sparse models considered in that proof. Consequently, the same type-I and type-II bounds continue to hold after replacing $\varepsilon$ by $\varepsilon-C_3\delta$ for a constant $C_3>0$. In particular, if $\varepsilon>c\delta$ with $c>C_3$, then the testing error probabilities remain bounded by $\exp(-Cn\varepsilon^2)$ for some $C>0$.

The final posterior bound follows from the same decomposition as in \citet{bai2018high}:
\[
\Pi(U_\varepsilon^c\mid \mathbf{Y},\mathbf{X})
\le
\Phi_n
+
(1-\Phi_n)
\frac{\int_{U_\varepsilon^c} p_{\mathbf{B}}(\mathbf{Y}\mid \mathbf{X})\pi(\mathbf{B})\,d\mathbf{B}}
{\int p_{\mathbf{B}}(\mathbf{Y}\mid \mathbf{X})\pi(\mathbf{B})\,d\mathbf{B}},
\]
where $U_\varepsilon=\{\mathbf{B}:\|\mathbf{B}-\mathbf{B}_0\|_F\le \varepsilon\}$. The numerator is controlled by the exponentially consistent tests, and the denominator is controlled by the TPBN prior mass around $\mathbf{B}_0$. This yields the stated exponential contraction bound.
\end{proof}

\begin{corollary}[Restatement of Corollary~\ref{cor:euclid} in the main text]
Under Euclidean distance with centered $\mathbf{X}$, suppose $\log p=o(n)$ and Assumptions~A and B hold. Then, for every $\varepsilon>0$,
\[
\sup_{\mathbf{B}_0}
\E_{\mathbf{B}_0}
\Pi\!\left(
\mathbf{B}:
\|\mathbf{B}-\mathbf{B}_0\|_F
>
\varepsilon\left(\frac{\log p}{n}\right)^{1/2}
\;\middle|\;
\mathbf{Y}=\mathbf{Z}_k,\mathbf{X}
\right)
\longrightarrow 0
\qquad\text{as}\qquad \min\{n,p\}\to\infty.
\]
\end{corollary}

\begin{proof}[Proof of Corollary~\ref{cor:euclid}.]
When the dissimilarity is Euclidean and the columns of $\mathbf{X}$ are centered, $\mathbf{G}=\mathbf{X}\mathbf{X}^\top$ and the leading PCoA coordinates satisfy $\mathbf{Z}_k=\mathbf{X}\mathbf{V}_k$, where $\mathbf{V}_k$ contains the first $k$ PCA loading vectors. Hence $\delta=0$, $\mathbf{R}_0=\mathbf{0}$, and the pseudo-response regression is exactly the multivariate regression model studied by \citet{bai2018high}. The stated $(\log p/n)^{1/2}$ rate therefore follows directly from their Euclidean-case result.
\end{proof}

\bigskip
The proof above is intentionally concise. It clarifies how the $\delta$-surrogate enters the argument and isolates the only new ingredient relative to the exact multivariate regression theory of \citet{bai2018high}, namely the deterministic bias bound in \eqref{eq:biasbound}.

\section{Additional Euclidean-case experiments}\label{app:euclid}
\resetappendixcounters

This appendix reports additional experiments in Euclidean settings. They serve two purposes: first, to verify that BSPCoA reduces to Bayesian sparse PCA when the dissimilarity is Euclidean; and second, to retain the large-scale runtime comparison while keeping the main paper focused on distance-based microbiome ordination. Because principal components are only defined up to sign, componentwise sign reversals in the loading tables are immaterial.

\subsection*{B.1. Euclidean latent-factor simulations}

We generated data from a two-factor model similar to the settings in \citet{guan2009} and \citet{zou2006sparse}. Suppose $n$ observations on $p$ variables are available. The first $p_1$ variables are associated with a latent factor $V_1$, the next $p_2-p_1$ variables are associated with another latent factor $V_2$, and the remaining variables are pure noise. Specifically,
\[
X_j=
\begin{cases}
V_1+\varepsilon_j, & j=1,\ldots,p_1,\\
V_2+\varepsilon_j, & j=p_1+1,\ldots,p_2,\\
\varepsilon_j, & j=p_2+1,\ldots,p,
\end{cases}
\qquad
\varepsilon_j\overset{\mathrm{iid}}{\sim} N(0,1).
\]
In the settings below, $V_1\sim N(0,10)$, $V_2\sim N(0,20)$, $n=50$, and the remaining design parameters are indicated in the table captions.

\begin{table}[H]
\centering
\caption{Euclidean latent-factor simulation with $n=50$, $p=10$, $p_1=1$, and $p_2=4$. The Bayesian method recovers the same sparse support as SPCA up to componentwise sign changes.}
\label{tab:euclid_small}
\small
\begin{tabular}{lccccccc}
\toprule
 & \multicolumn{3}{c}{PCA} & \multicolumn{2}{c}{SPCA} & \multicolumn{2}{c}{Bayesian SPCA} \\
\cmidrule(lr){2-4}\cmidrule(lr){5-6}\cmidrule(lr){7-8}
Variable & PC1 & PC2 & PC3 & PC1 & PC2 & PC1 & PC2 \\
\midrule
$X_1$ & 0.294 & 0.914 & 0.026 & 0.000 & 1.000 & 0.000 & 1.000 \\
$X_2$ & 0.531 & -0.168 & -0.253 & 0.578 & 0.000 & -0.578 & 0.000 \\
$X_3$ & 0.529 & 0.168 & 0.221 & 0.574 & 0.000 & -0.573 & 0.000 \\
$X_4$ & 0.527 & -0.177 & -0.221 & 0.578 & 0.000 & -0.578 & 0.000 \\
$X_5$ & -0.023 & 0.029 & 0.297 & 0.000 & 0.000 & 0.000 & 0.000 \\
$X_6$ & -0.023 & -0.028 & -0.276 & 0.000 & 0.000 & 0.000 & 0.000 \\
$X_7$ & -0.019 & 0.027 & 0.292 & 0.000 & 0.000 & 0.000 & 0.000 \\
$X_8$ & 0.021 & -0.028 & -0.287 & 0.000 & 0.000 & 0.000 & 0.000 \\
$X_9$ & 0.023 & 0.029 & 0.290 & 0.000 & 0.000 & 0.000 & 0.000 \\
$X_{10}$ & 0.019 & 0.029 & 0.314 & 0.000 & 0.000 & 0.000 & 0.000 \\
\midrule
AVar (\%) & 53.8 & 33.0 & 2.9 & 51.2 & 34.5 & 51.2 & 34.4 \\
\bottomrule
\end{tabular}
\end{table}

\begin{table}[H]
\centering
\caption{Euclidean latent-factor simulation with $n=50$, $p=100$, $p_1=3$, and $p_2=8$. Only the first 10 variables are shown.}
\label{tab:euclid_large}
\small
\begin{tabular}{lccccccc}
\toprule
 & \multicolumn{3}{c}{PCA} & \multicolumn{2}{c}{SPCA} & \multicolumn{2}{c}{Bayesian SPCA} \\
\cmidrule(lr){2-4}\cmidrule(lr){5-6}\cmidrule(lr){7-8}
Variable & PC1 & PC2 & PC3 & PC1 & PC2 & PC1 & PC2 \\
\midrule
$X_1$ & 0.0554 & -0.5554 & -0.0609 & 0.0000 & -0.5789 & 0.0000 & -0.5766 \\
$X_2$ & 0.0522 & -0.5512 & -0.0701 & 0.0000 & -0.5728 & 0.0000 & -0.5673 \\
$X_3$ & 0.0555 & -0.5559 & -0.0684 & 0.0000 & -0.5767 & 0.0000 & -0.5758 \\
$X_4$ & 0.4417 & -0.0431 & -0.0766 & 0.4493 & 0.0000 & 0.4510 & 0.0000 \\
$X_5$ & 0.4420 & -0.0464 & -0.0732 & 0.4462 & 0.0000 & 0.4484 & 0.0000 \\
$X_6$ & 0.4406 & 0.0437 & 0.0720 & 0.4479 & 0.0000 & 0.4479 & 0.0000 \\
$X_7$ & 0.4375 & -0.0431 & -0.0624 & 0.4433 & 0.0000 & 0.4368 & 0.0000 \\
$X_8$ & 0.4378 & -0.0453 & -0.0784 & 0.4466 & 0.0000 & 0.4392 & 0.0000 \\
$X_9$ & -0.0112 & -0.0206 & -0.0710 & 0.0000 & 0.0000 & 0.0000 & 0.0000 \\
$X_{10}$ & 0.0128 & -0.0209 & -0.0809 & 0.0000 & 0.0000 & 0.0000 & 0.0000 \\
\midrule
AVar (\%) & 45.38 & 14.20 & 2.34 & 44.11 & 13.42 & 43.83 & 13.35 \\
\bottomrule
\end{tabular}
\end{table}

\subsection*{B.2. Runtime comparison for the subsample projection strategy}

To illustrate the computational advantage of the subsample projection strategy, we compared the runtime of classical PCoA with BSPCoA as the sample size increased. For BSPCoA, the principal directions were estimated from a fixed subsample of size $m=100$, after which the remaining observations were embedded by linear projection. The results in Figure~\ref{fig:runtime_comparison} show that the runtime of classical PCoA increases rapidly with $n$, whereas the subsample-based BSPCoA procedure grows much more slowly. Figure~\ref{fig:runtime_clustering} shows the clustering results of classical PCoA and BSPCoA under different sample sizes. The results indicate that BSPCoA maintains comparable clustering performance while requiring substantially less computational time than classical PCoA.

\begin{figure}[H]
\centering
\safeincludegraphics[width=0.9\textwidth]{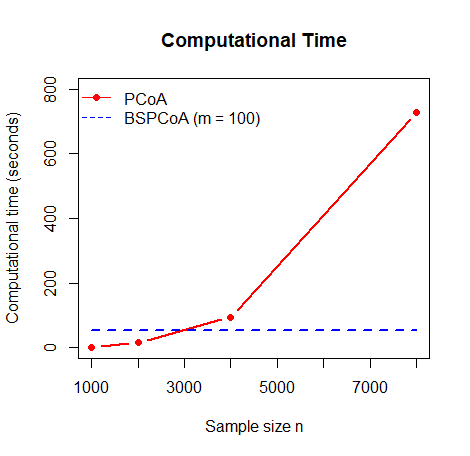}
\caption{Computational time comparison between classical PCoA and BSPCoA as the sample size increases. BSPCoA uses a fixed subsample of size $m=100$, so the distance-based estimation step remains nearly constant while classical PCoA becomes progressively more expensive.}
\label{fig:runtime_comparison}
\end{figure}

\begin{figure}[p]
\centering
\begin{subfigure}{0.48\textwidth}
\centering
\safeincludegraphics[width=\linewidth]{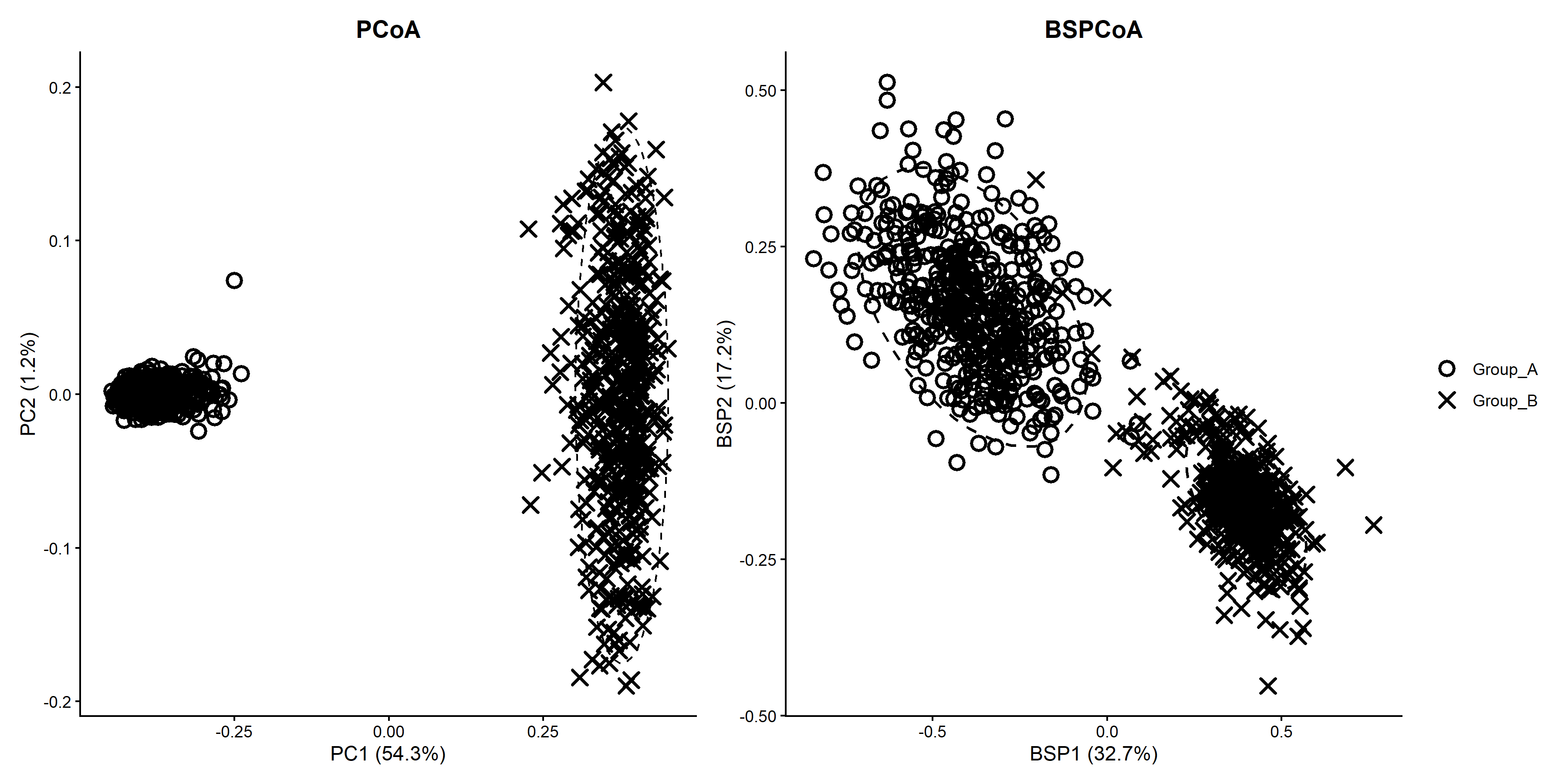}
\caption{$n=1000$}
\end{subfigure}\hfill
\begin{subfigure}{0.48\textwidth}
\centering
\safeincludegraphics[width=\linewidth]{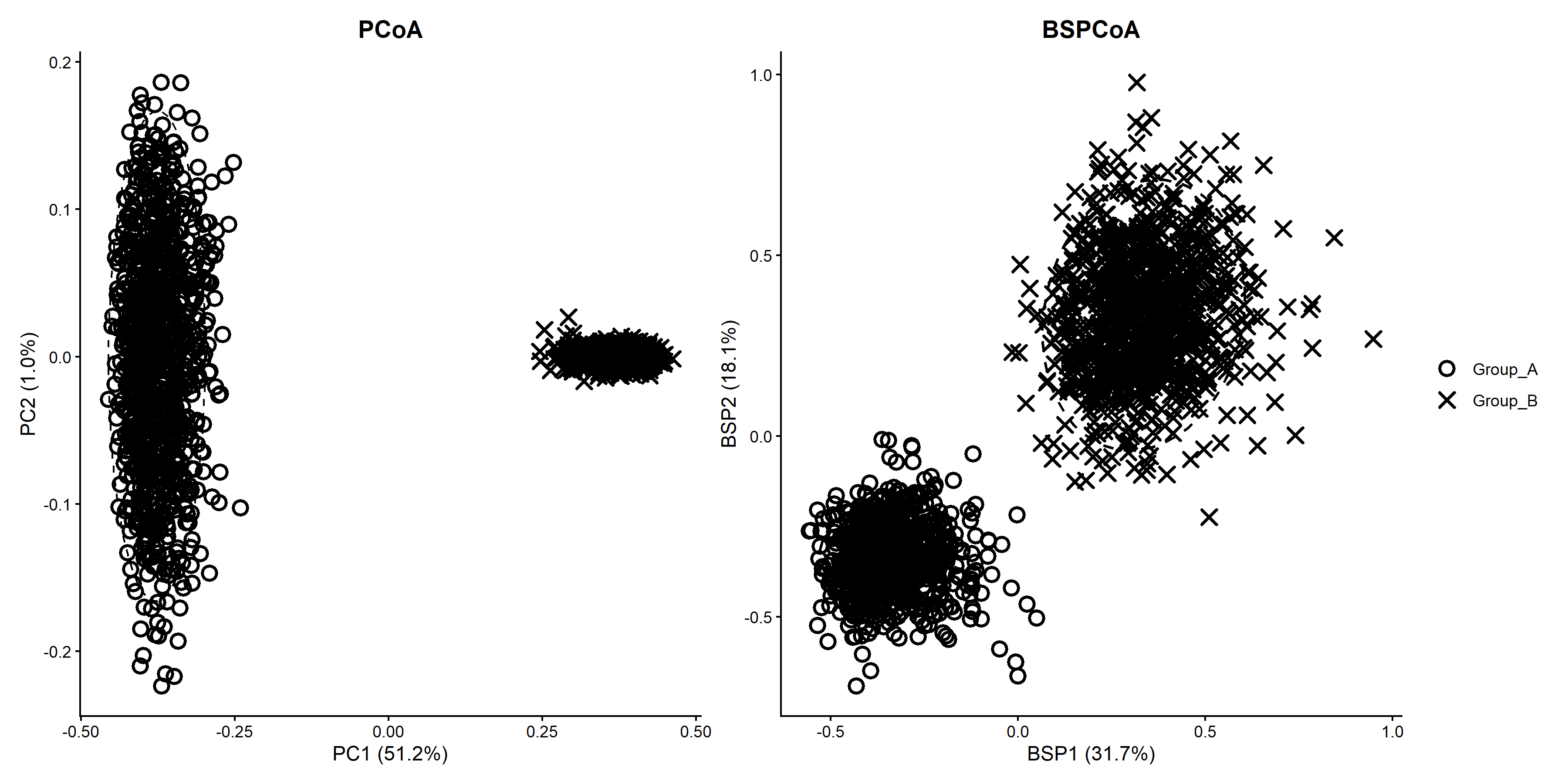}
\caption{$n=2000$}
\end{subfigure}

\medskip

\begin{subfigure}{0.48\textwidth}
\centering
\safeincludegraphics[width=\linewidth]{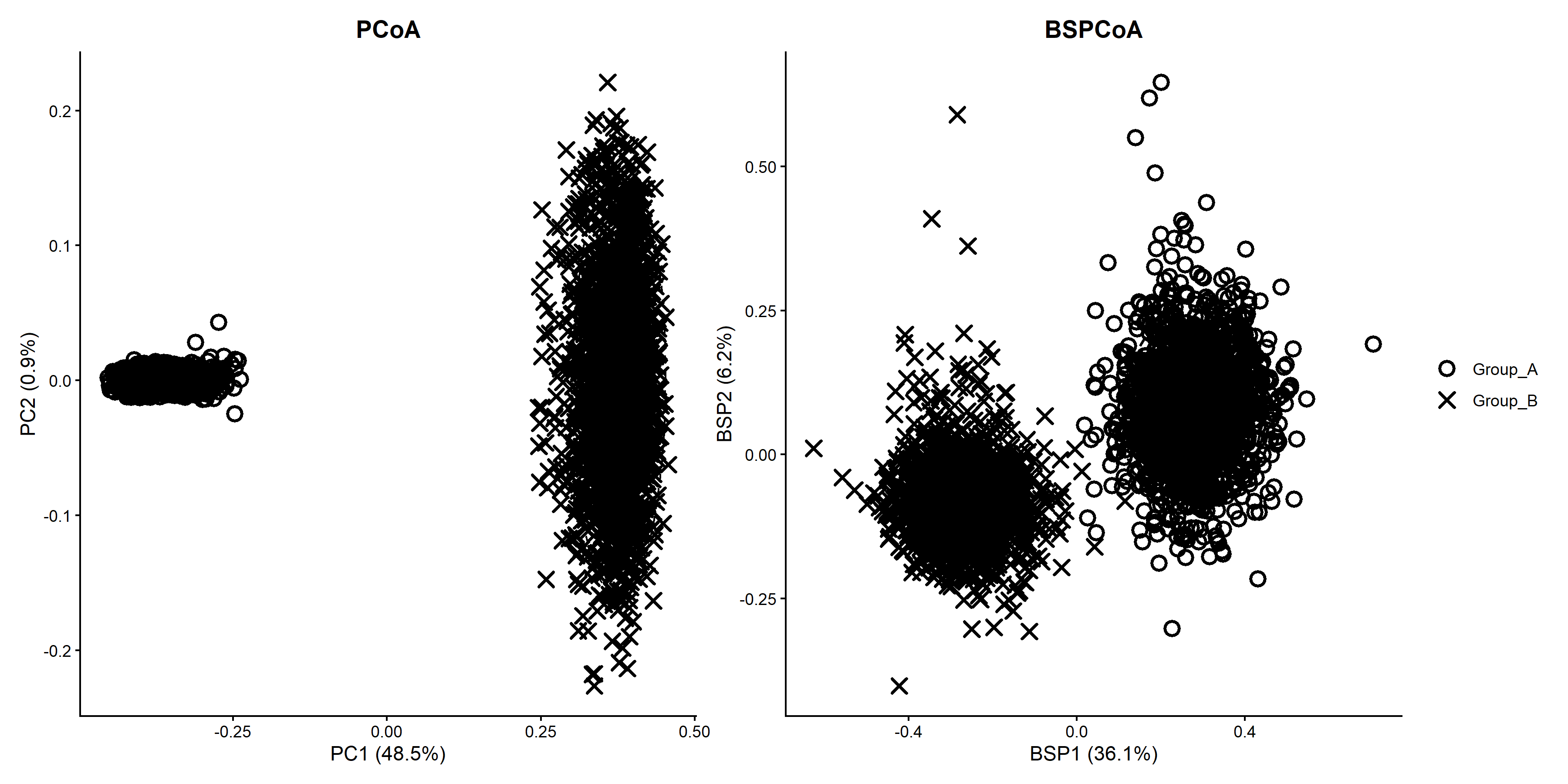}
\caption{$n=4000$}
\end{subfigure}\hfill
\begin{subfigure}{0.48\textwidth}
\centering
\safeincludegraphics[width=\linewidth]{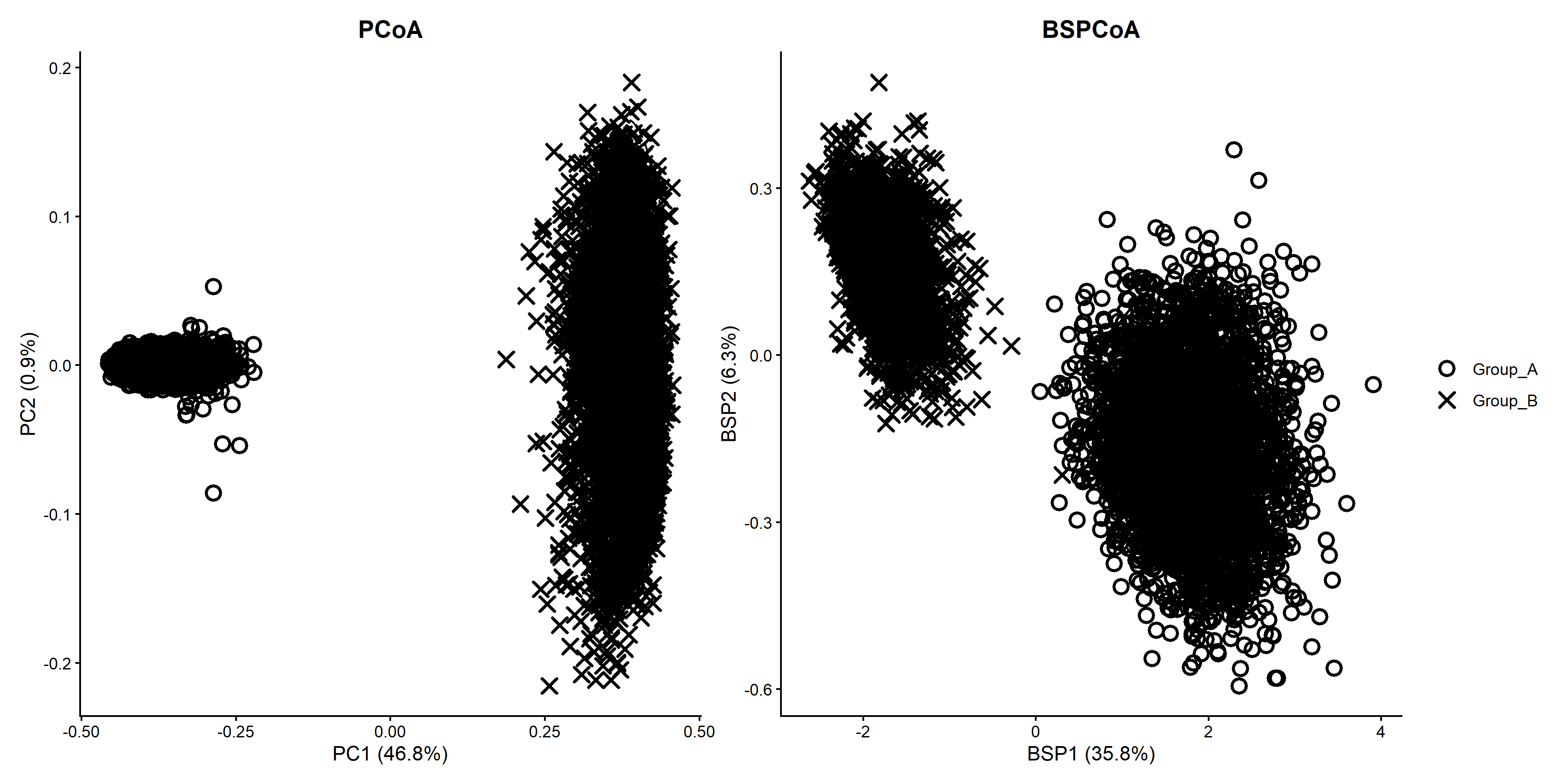}
\caption{$n=8000$}
\end{subfigure}
\caption{Clustering comparison between classical PCoA and BSPCoA under different sample sizes. The BSPCoA clustering results are obtained using the estimated projection matrix derived from a fixed subsample of size $m=100$.}
\label{fig:runtime_clustering}
\end{figure}


\begin{thebibliography}{99}

\bibitem[Anderson and Willis(2003)]{anderson2003}
Anderson, M.~J. and Willis, T.~J. (2003).
Canonical analysis of principal coordinates: A useful method of constrained ordination for ecology.
\emph{Ecology} \textbf{84}, 511--525.

\bibitem[Armagan et~al.(2011)Armagan, Clyde, and Dunson]{armagan2011generalized}
Armagan, A., Clyde, M., and Dunson, D.~B. (2011).
Generalized beta mixtures of Gaussians.
In \emph{Advances in Neural Information Processing Systems 24}, 523--531.

\bibitem[Bai(2018)]{bai2018high}
Bai, R. (2018).
\emph{Bayesian High-Dimensional Models with Scale-Mixture Shrinkage Priors}.
PhD dissertation, University of Florida.

\bibitem[Bai and Ghosh(2018)]{bai2018mbsp}
Bai, R. and Ghosh, M. (2018).
High-dimensional multivariate posterior consistency under global-local shrinkage priors.
\emph{Journal of Multivariate Analysis} \textbf{167}, 157--170.

\bibitem[Carvalho et~al.(2010)Carvalho, Polson, and Scott]{carvalho2010}
Carvalho, C.~M., Polson, N.~G., and Scott, J.~G. (2010).
The horseshoe estimator for sparse signals.
\emph{Biometrika} \textbf{97}, 465--480.

\bibitem[Gower(1966)]{gower1966some}
Gower, J.~C. (1966).
Some distance properties of latent root and vector methods used in multivariate analysis.
\emph{Biometrika} \textbf{53}, 325--338.

\bibitem[Guan and Dy(2009)]{guan2009}
Guan, Y. and Dy, J. (2009).
Sparse probabilistic principal component analysis.
In \emph{Proceedings of the 12th International Conference on Artificial Intelligence and Statistics}, 185--192.

\bibitem[Legendre and Anderson(1999)]{legendre1999}
Legendre, P. and Anderson, M.~J. (1999).
Distance-based redundancy analysis: Testing multispecies responses in multifactorial ecological experiments.
\emph{Ecological Monographs} \textbf{69}, 1--24.

\bibitem[Legendre and Gallagher(2001)]{legendre2001}
Legendre, P. and Gallagher, E.~D. (2001).
Ecologically meaningful transformations for ordination of species data.
\emph{Oecologia} \textbf{129}, 271--280.

\bibitem[Lin and Fong(2019)]{lin2018}
Lin, L. and Fong, D.~K.~H. (2019).
Bayesian multidimensional scaling procedure with variable selection.
\emph{Computational Statistics \& Data Analysis} \textbf{129}, 1--13.

\bibitem[Oh and Raftery(2001)]{oh2001}
Oh, M.~S. and Raftery, A.~E. (2001).
Bayesian multidimensional scaling and choice of dimension.
\emph{Journal of the American Statistical Association} \textbf{96}, 1031--1044.

\bibitem[Smits et~al.(2017)Smits, Leach, Sonnenburg, et~al.]{smits2017seasonal}
Smits, S.~A., Leach, J., Sonnenburg, E.~D., et~al. (2017).
Seasonal cycling in the gut microbiome of the Hadza hunter-gatherers of Tanzania.
\emph{Science} \textbf{357}, 802--806.

\bibitem[Wang et~al.(2025)Wang, Bai, and Huang]{wang2025ejs}
Wang, S.-H., Bai, R., and Huang, H.-H. (2025).
Two-step mixed-type multivariate Bayesian sparse variable selection with shrinkage priors.
\emph{Electronic Journal of Statistics} \textbf{19}, 397--457.

\bibitem[Zou et~al.(2006)Zou, Hastie, and Tibshirani]{zou2006sparse}
Zou, H., Hastie, T., and Tibshirani, R. (2006).
Sparse principal component analysis.
\emph{Journal of Computational and Graphical Statistics} \textbf{15}, 265--286.

\end{thebibliography}
\end{document}